\documentclass[manuscript,nonacm]{acmart}

\usepackage{hyperref}
\usepackage{url}
\usepackage{color}
\usepackage{enumerate}
\usepackage{subcaption}

\newcommand{\I}{\ensuremath{\mathcal{I}}\xspace} 
\newcommand{\ca}{\ensuremath{\mathtt{CA}}\xspace}
\newcommand{\statcons}{\ensuremath{\mathtt{Consensus}_{Stat}}\xspace}

\newcommand{\PS}{\ensuremath{\mathcal{P}}\xspace}

\newcommand{\party}{\ensuremath{p}\xspace}

\newcommand{\wc}{\prot{WeakConsensus}}%
\newcommand{\gc}{\ca\xspace}%

\newcommand{\psetp}[2]{\ensuremath{P_{{#1}}^{_{(#2)}}}\xspace}%
\newcommand{\nv}{{\mbox{\ensuremath{\text{``{\sf\small n/v}''}}}}\xspace}%
\newcommand{\pa}[1]{#1}%
\newcommand{\prot}[1]{\ensuremath{\mathsf{#1}}\xspace}%

\newenvironment{protocol}[2][Protocol]{%
   \setlength{\tabcolsep}{.3em}
  \begin{center}
    \begin{tabular}{|l|}\hline
      \begin{minipage}{.975\linewidth}
       \vspace{.5ex}\subsubsection*{#1 #2} \begin{enumerate}\setlength{\itemsep}{0mm} \setlength{\parskip}{3pt}
  \setlength{\parsep}{0pt}
\vspace{-.5ex}
        }{%
        \end{enumerate}\vspace{-1.2ex}
      \end{minipage}\\ 
      \hline
    \end{tabular}
  \end{center}
}

\newenvironment{fprotocol}[2][Protocol]{%
\begin{figure}[ht!]\normalsize
\setlength{\tabcolsep}{.3em}
\hspace{.16ex}\begin{tabular}{|l|}\hline
      \begin{minipage}{.975\linewidth}\vspace{0.5ex}
        \subsubsection*{#1 #2} \begin{enumerate}\setlength{\itemsep}{0mm} \setlength{\parskip}{2pt}
  \setlength{\parsep}{0pt}
\vspace{-.5ex}
        }{%
        \end{enumerate}
        \vspace{-1.2ex} 
      \end{minipage}\\
      \hline
    \end{tabular}\vspace{-2ex}
\end{figure}
}

\newcommand{\cancel}[1]{}

\usepackage[boxed, ruled,noend,linesnumbered]{algorithm2e}

\usepackage{wrapfig}
\usepackage{natbib}
\newtheorem{remark}{Remark}

\newtheorem{theorem}{Theorem}
\newtheorem{corollary}[theorem]{Corollary}
\newtheorem{definition}{Definition}

\newcounter{linenumber}

\def\R{\ensuremath{\mathcal{R}}}

\def\I{\ensuremath{\mathcal{I}}}
\def\O{\ensuremath{\mathcal{O}}}

\newcommand{\remove}[1]{}

\def\O {\mathcal{O}}
\def\I {\mathcal{I}}

\title{Synchrony/Asynchrony vs. Stationary/Mobile? The Latter is Superior...in Theory. }
\author{Eli Gafni}
\affiliation{\institution{UCLA}}
\email{eli@ucla.edu}
\author{Vasileios Zikas}
\affiliation{\institution{Purdue University}}
\email{vzikas@purdue.edu}
\date{}

\begin{document}

\begin{abstract}
Like Asynchrony, \emph{Mobility} of faults precludes consensus. Yet, a model $M$ in which Consensus
is solvable, has an analogue relaxed model in which Consensus is not solvable and for which we can ask, whether Consensus is solvable if the system initially behaves like the relaxed analogue model, but eventually morphs into $M$. We consider two relaxed analogues of $M$.
The first is the traditional Asynchronous model, and the second to be defined, the Mobile analogue. While for some $M$ we show that Consensus is not solvable in the Asynchronous analogue, it is solvable in all the Mobile analogues. Hence, from this perspective Mobility is superior to Asynchrony.  

The pie in the sky relationship we envision is: Consensus is solvable in $M$, if and only if binary Commit-Adopt is solvable in the mobile analogue. 

The ``only if'' is easy. Here we show case by case that the ``if'' holds for all the common faults types.

\end{abstract}

\maketitle

\begin{acks}
This version was submitted to PODC 2020.
\end{acks}
\section{Introduction}

\subsection{Mobility and Message Adversary}

The notion of {\emph{indulgence}, a term usually frowned upon, has found merit in distributed computing as an adjective for an algorithm that solves a task in an environment which is initially asynchronous but eventually behaves synchronously \cite{PODC:Guerraoui00}. An algorithm indulgences the asynchronous period
in the sense of preserving safety. It becomes live in the times of synchrony. Obviously this notion can be extended to a task $T$, by saying that $T$ is indulgent if there exist an indulgent algorithm to solve $T$. More precisely, a task $T$ is indulgent if when $T$ is solvable in any  model with certain type and number of faults, it is also solvable eventually. 

Thus, indulgence takes a task and all the models in which the task is solvable. For each model we assume a relaxed model is defined. Then the task is indulgent if it is solvable 
 when the system starts in the respective relaxed model, but eventually behaves as the (unrelaxed) model.



Thus indulgence has there parameters: The Task $T$, the fault type, and the pairs of unrelaxed/relaxed models.

Here we consider the task to be Consensus, and the pairs to be Synchronous/Asynchronous for the common fault types. 

This paper's motivation is the displeasing observation that for two fault types consensus is not indulgent: 
\begin{enumerate}
    \item Send-Omission Faults: Consensus in the synchronous case is achievable for $t<n$, where $t$ is the number of faults and $n$ is the number of processors. In contrast, it is a folklore that at the time of asynchrony no safety can be maintained for $t \geq n/2$ as network partition occurs.
    \item Authenticated Byzantine: Consensus in the synchronous case is achievable for $t<n/2$~\cite{DolStr83,Fit02}. In contrast, at the time of asynchrony no safety can be maintained for $t \geq n/3$~\cite{PODC:DLS84}.
\end{enumerate}

To avoid displeasure, we investigate replacing the Synchrony/Asynchrony pair with the Stationary/Mobile pair. The system is synchronous, each processor gets a signal of an end of a round after which no message is in transit to it. The misbehavior is when the faults are mobile,
and the desired behavior is when the faults become stationary. If the number of fault is $t$, in the stationary case only fixed $t$ processors can experience faults. In the mobile case, each round $t$ different processors may exhibit faulty behaviour.

Usually, faults are attributed to processors. Here we restrict ourselves to an adversary which attacks a processor by controlling their messages sending interface. Thus a 1-omission resilient mobile system will be a \emph{synchronus} system where all send to all and in each round an adversary chooses a processor and can drop any number of the messages the processor sends (cf. \cite{DBLP:journals/ipl/JayantiCT99}). In the Byzantine case, the adversary can not only drop messages but also tamper with them. And finally, we attend to the Authenticated Byzantine, which posed the foremost technical difficulty of defining the notion of mobility, and proving the desired result for it.

When defining authentication, one often thinks of its cryptographic instantiation by means of digital signature. But authentication has an English-Dictionary definition independent of signatures.
It is about verification that a claim holds. the verification of a claim that a painting is by Picasso, or that a fossil is from the Paleolithic Period, or that a Diary was written by Hitler, is called Authentication . Thus in this paper Authenticated Byzantine is an abstract assumptions on which claims can be verified (and consequently not forged) and which cannot. Obviously in our case it will be about ``if processor $p_i$ claimed it received a message $m$ at round $j$ from $p_k$'' can the receiver of such a claim can verify verify whether the claim is true or not. Thus, it constraints the adversary, to make only claims that cannot be verified as false. Consequently, forgetting the means of verification, we consider authentication set of pairs $(p_i,j)$ where $p_i$ is a processor and $j$ is round number.
The messages sent by $p_i$ at round $j$ can be forged while message generated by pairs not in the set cannot be forged.

Our results in the Authenticated Byzantine is for this abstract constrains. To our knowledge, Authenticated Byzantine was never  defined for the mobile case at this clean level of abstraction. A challenge this paper poses is to find an implementation, e.g., by using cryptography under appropriate assumptions, of the functionality of our abstract definition. Nevertheless, proof can still be done on the abstract functionality level. 

Our first encounter with idea of mobility was through the beautiful observation of Santoro and Widmayer \cite{STACS:SW89} that the FLP proof of consensus impossibility translates verbatim to the mobile setting, and afterwards for privacy-violating corruptions in the context of  \emph{proactive-security}~\cite{PODC:OstYun91}. 

It isn't clear what Santoro and Widmayer had in mind as the cause of faults when they talked about omission faults. Why did the omission faults occur?
There can be two views for that. One is that the processor misbehaved, the other is that an adversary intercepted messages sent, and dropped them. Which view one chooses, conceptually makes a big difference. In the Byzantine failures case the generalization of the former is that a virus got control of the processor. The generalization of the latter is that a deamon only tampered with messages. 

Theoretically speaking, the latter view of a deamon, rather than a virus---which in fact gave rise to the notion of \emph{message-adversary} (MAd) by Afek and Gafni \cite{}---is more pleasing: In this approach a message adversary unified many seemingly unrelated notions. For instance, MAd when generalized here to the Byzantine, has processors being always innocent and good. Thus we avoid questions like ``are corrupt processors required to output and, if so, what should their output be?'' or specifying a problem in terms of ``correct'' and ``incorrect'' which are notions associated with an execution, defeating the presentation of task as solely a mathematical relation, independent of the environment in which it is to be solved. The analogue of a function in centralized computing. 

Alas, thinking of the adversary, especially, in the Byzantine case as being a Message Adversary may not sit well with reality. This paper isn't about reality, it is about whether we can get the mathematics to be nicer.

This paper is in the mold of \emph{set consensus} \cite{DBLP:conf/podc/Chaudhuri90}. There is no \emph{killer application} for set consensus. There might never be. Nevertheless, the mathematics say that for distributed computing there is no preference to consensus over set consensus. Similarly here. Yet, in fact, some solutions for proactive security move in our theoretical direction of MAd \cite{PODC:OstYun91}.

\subsection{Consensus vs. Commit-Adopt}

Commit-Adopt (CA) \cite{DBLP:conf/podc/Gafni98} is a task, relaxing consensus. It requires ``liveness'' only when all processors start with the same value. Otherwise it should be safe, if one processor outputs, all adopt that output value to be their new input value. 

This paper brings to the fore the informal conjecture that Consensus is to stationarity what CA is to mobility. One could have seen a glimpse of this in the theoretically tainted pair Synchrony/Asynchrony. In the Byzantine case it is Consensus vs. Realiable Broadcast \cite{DBLP:journals/iandc/Bracha87}.
And what is Reliable Broadcast if not shared-memory (appropriately defined) and consequently CA (after some thought).The domain of models Synchronous/Asynchronous for which the pair Consensus/Reliable Broadcast are solvable, coincide i.e. $<n/3$.

There are two direction to the above conjecture. The easy direction is that for any mobile setting in which binary CA is solvable, then consensus is solvable in the stationary analogue. This follows from the construction in \cite{FOCS:BerGarPer89,FitMau98} of a consensus algorithm effectively made out of CAs. The reverse direction is the heart of the conjecture. It for instance, should imply: Take any Consensus algorithm for the stationary setting, if executed in the analogue mobile setting and all processors start with the same bit $b$, then they should all output $b$. We were not as yet successful in proving this.

Rather, we went through case by case and showed the following: For all common fault types, the  number of faults Consensus can tolerate in the Stationary case is the number of faults  CA can tolerate in the analogue mobile case.

\section{Model}

\subsection{Consensus as a Task}
A \emph{task} on $n$ processors is defined as a binary relation $R\subseteq\I^n\times\O^n$  between $n$ tuples where entries for the first tuple come from an \emph{input space} $\I$ and the second tuple, the outputs come from an {\em output space} $\O$. The interpretation of each $(x_1,\ldots,x_n)\in \I^n$ (resp. $(y_1,\ldots,y_n)\in\O^n$) is that  $x_i$ (resp. $y_i$) is the input (resp. output) value of $p_i$.

For the  consensus task, the corresponding relation, $R_C$, satisfies the following, where every processor $\party_i$ has input $v_i$: 
\begin{enumerate}
    \item If a processor outputs $v$ then some processor must have  input $v$
    \item If a processor outputs $v$ then no processor outputs $v'\neq v$. 
\end{enumerate}

In this work we restrict our attention to {\em binary consensus}---i.e., the inputs and output are bits.

\cancel{
Defined as a task, in binary consensus the input tuples are all the possible $2^n$ binary vectors of size $n$, i.e., $\I=\{0,1\}$.
There are only two output tuples, the all-$0$ tuple and the all-$1$ tuple. For notational simplicity we will denote the all-$b$ tuple by $b^n$.
The all-$b$'s input tuple $b^n$ relates only to the output tuple $b^n$, respectively. All other tuples relate to both possible output tuples.

\begin{definition}[Consensus]
Binary consensus as a task is defined by the following binary relation $R_C\subseteq\{0,1\}^n\times\{0,1\}^n$: 

$$\R_C := \{(0^n,0^n),(1^n,1^n)\}\cup\{(x,y) \text{ s.t } x\in\{0,1\}^n\setminus\{0^n,1^n\}\} \text{ and } y\in\{0^n,1^n\}\}  $$

\end{definition}}

Note that the above definition of consensus is different from the traditional consensus definition in the Byzantine setting~\cite{PSL80,LSP82}, which does not consider the inputs and gives no guarantee on the outputs of corrupted/Byzantine parties. 

An {\em interactive processor} is a processor that in addition to local computation, has the ability to communicate with other processors by sending them messages over some network. A {\em protocol} among $n$ processors is a collection of $n$ interactive processor specifications.
\footnote{We restrict our attention to non-reactive tasks where processors receive only one input at the beginning and produce a single local output (see below). This type of protocols is sufficient for our results; however, one can extend this definition to reactive tasks.}  

A {\em model} $M$ consists of two components, the {\em  system model} and the {\em communication model}. The system model specifies the types and capabilities of the processors, along with the properties of the communication between them. 
The {\em failure (aka adversary) model} specifies the types, e.g., fail-stop, omission, byzantine, etc, and combination, e.g., up to $t$ faults, of different faults that the processors might endure. 

A task $T$ with corresponding relation $R_T$ is {\em solvable} in a model $M$ if there exists a protocol in $M$ such that if processors start with some input tuple, they all output from an output tuple relating, through $R_T$ , to the input tuple.

\begin{definition} 
An $n$-processor task $T$ described by relation $\R_T\subseteq \I^n\times\O^n$ is {\em solvable in a model $M$} if there exists a protocol $\Pi$ in $M$ such that $\forall (x_1,\ldots,x_n)\in\I^n$, if every processor $\party_i$ runs (his code of) $\Pi$ on input $x_i$, then every processor outputs $y_i\in\O$ such that $\left((x_1,\ldots,x_n),(y_1,\ldots,y_n)\right)\in \R_T$.
\end{definition}

\paragraph{Commit-Adopt}

In addition to consensus we will use Commit-Adopt (CA) task~\cite{DBLP:conf/podc/Gafni98}. CA   (also referred to as {\em graded consensus}~\cite{TCC:ZikHauMau09})  is similar in spirit to gradecast~\cite{STOC:FelMic88,DBLP:conf/wdag/Ben-OrDH10,PODC:Aspnes10}. 


Its specification is as follows:

\noindent Every processor holds as input a  value $v_i$.  Every  processor $\party_j$ outputs a value which is either  $commit(v_j)$ or $adopt(v_j)$ for some value $v_j$, which equals the input of some processor.  
\begin{enumerate}
    \item If all processors start with the same value $v_i=v$, then they all output $commit(v)$,
    \item If a processor outputs $commit(v)$, then  all processors output either $commit(v)$, or $adopt(v)$. 
\end{enumerate}

In this work we focus our attention to {\em binary CA} where the input and output values $v$ are bits. 

\cancel{
\begin{definition}[Commit-Adopt~\cite{\cite{DBLP:conf/podc/Gafni98}}]
Binary commit-adopt as a task is defined by the following  relation $R_C\subseteq\{0,1\}^n\times\{commit(0),adopt(0),commit(1),adopt(1)\}^n$:

\[
\begin{split}\R_C :=& \biggl\{(0^n,commit(0)^n),(1^n,commit(1)^n)\biggr\}\cup\\
&\biggl\{(x,y) \text{ s.t } x\in\{0,1\}^n\setminus\{0^n,1^n\}\\
&\text{ and } y\in\{(y_1,\ldots,y_n)\in\{commit(0),adopt(0),commit(1),adopt(1)\}^n,
\\
&\text{ s.t. if for some } (i,b)\in[n]\times\{0,1\}\  y_i=commit(b) \text{ then } \forall j\in[n]\  y_j\in\{commit(b),adopt(b)\}\}\biggr\}
\end{split} 
\]
\end{definition}}

In the remainder of this section, we describe the model(s) under consideration. All models considered here share the same system-model   component, but are for different adversary models. 

\subsection{The System model}
We assume $n$ interactive processors $\PS=\{\party_1,\party_2, \ldots , \party_n\}$, also referred to as {\em parties}, which can be interactive computing machines, e.g., Interactive Turing Machines (ITM). As usual we will assume that the processors can perform polynomially long (in $n$) computations (and can communicate polynomially long messages), although our negative results even hold for unbounded processors.  
The processors are 
connected by a complete reliable communication-network with a dedicated channel between any two processors, a la LSP\cite{LSP82,PSL80}. Here, reliable means that if $\party_j$ receives a message from $\party_i$ (on the dedicated channel connecting the two processors) then $\party_j$ knows that this message was indeed sent by $\party_i$, or by the adversary on behalf of $\party_i$---if $\party_i$ was corrupted when the message was sent. We assume that the protocol (and communication) is {\em synchronous}. In particular, \ 
the protocol proceed in rounds, where in  each round all processors send a (potentially different) message to all other processors. All messages sent in any round $\rho$ are delivered by the beginning of round $\rho+1$.

\subsection{The (MAd-)Adversary Model}
Recall that our definition of a protocol solving a task requires even corrupted parties/processors to output a value and we give them the same output guarantees (e.g., agreement in the context of consensus) as we give to honest/uncorrupted parties. This makes our corruption model more suitable to capture a mobile adversary who in the course of the protocol might corrupt every processor. 

To be able to achieve such stronger guarantees we consider adversaries that operate at ``the network interface" (e.g., communication tapes) of the parties, rather than corrupting the parties' internal state. This allows us to define producing an output as writing it on a special write-only and append-only output  tape which is out of bounds for the adversary.  Here is how our adversary is defined.

We consider a central adversary who might affect  messages sent by  parties it corrupts; we refer to such an adversary as a {\em message adversary}, in short {\em MAd} (adversary). We note in passing that the notion of \emph{corrupted processor} at a round is just for descriptive means of delineating the power of the MAd Adversary. More concretely, a MAd adversary might intercept the {\em outgoing messages}, rather than the internal state, of processors. In a nutshell, in every round, each corrupted party prepares its messages for the current round, according to the messages received from previous rounds, by following his protocol instructions; depending on the privacy assumption on the underlying communication model, a MAd adversary   can tamper with theses messages.


To make the strongest possible statements, here we will consider the full information model~\cite{DBLP:journals/jacm/FredericksonL87} whereby the adversary gets to see all messages exchanged in the protocol. We will assume that, subject to its constraints, the adversary is non-deterministic, and can produce either garbage or the set of messages that will be most detrimental to an operation of an algorithm. So rather than describing the adversary as  an algorithm operating in its history, we will describe simply which messages {\em cannot} be non-deterministically produced given a corruption pattern.   As an example, unforgeability of signatures for keys inaccessible to the adversary, e.g., keys of an uncorrupted party $\party_i$, would correspond to restricting the messages that the adversary might be able to inject to the protocol on behalf of $\party_i$ to those that are actually generated by $\party_i$. 
%
%
%

The adversary is described by means of when parties are corrupted---{\em stationary, mobile}, or {\em  eventually stationary (mobile)}---and by the corruption types---{\em omission, Byzantine,} and {\em Authenticated Byzantine}, as discussed  below.  

\paragraph{The stationary MAd adversary}
A stationary $t$-MAd Adversary is an adversary that can corrupt at most $t$ processor throughout an execution of a protocol.\footnote{This includes both static and adaptive corruptions as defined in the cryptographic literature.} 

\paragraph{The mobile MAd adversary} The {\em mobile} $t$-MAd adversary is restricted to corrupting at most $t$ processors in a \emph{round}. Thus over time all processor may experience message tampering albeit in different round. 


\paragraph{The eventually-stationary (mobile) MAd adversary} This is an adversary that for a finite (but unknown to the protocol) number $\rho$ of rounds behaves as a mobile adversary, but from round $\rho+1$ on becomes stationary.
More concretely, an {\em eventually-stationary  $t$-MAd adversary} is an adversary that plays a mobile $t$-MAd adversary strategy for a finite number of rounds and then chooses and from some point on confines its corruption to a fixed set of at most $t$ processors. 

\subsubsection{MAd-Corruption Types considered in the paper}
\cancel{\paragraph{Fail-Stop Corruption}
A processor $p_i$ is fail-stop corrupted if at some round $j$ the adversary removed some of its messages, and for all subsequent round $k$, $k>j$ the adversary removes all the messages it sends. Obviously, fail-stop has no mobile version.}
\paragraph{Send Omission failure}
The adversary is restricted to just removing messages but it is not constrained to remove all messages of a processor in rounds after it removed some. 
\paragraph{Byzantine failure}
The adversary can tamper with messages replacing them with any of its own choosing.
\paragraph{Authenticated Byzantine failure}
Like Byzantine only that if at round $j$ processor $p_i$ was not corrupted, then at subsequent rounds $k>j$ no processor can claim messages sent by $p_i$ at round $j$ anything that has not really been sent, aside from just pretending a message was not received. More formally, in the synchronous authenticated Byzantine setting we will assume wlog that any message sent in a protocol from 
$\party_i$ in round $\rho$ has the formal $(\party_i,\rho,m)$ where $m\in\{0,1\}^*$ is the contents of the message and $\party_i$ and $\rho$ are associated metadata. The authenticated byzantine adversary model mandates then that if for any message $(\party,\rho,m)$, where $p\in\PS$,
and for any party $\party'\in\PS$, $(\party',\rho',m')$ is  (an encoding of) a substring of  $m$, where $\rho'\leq\rho$, then either $\party'$ was corrupted in round $\rho'$ or he was uncorrected and sent $(\party',\rho',m')$.

\begin{remark}[Not giving up corrupted parties---even corrupted parties produce outputs]
One might consider a natural MAd-analogue of standard Byzantine and omission corruption to tamper with both incoming and outgoing communication. However, defining things this way leads to complications with how broadcast and consensus are defined. In particular, the traditional cryptographic definitions {\em give up} corrupted parties, i.e.,  give no guarantees about the output of corrupted parties. This means that not only these tasks cannot be defined as simple functions taking only processor's inputs into account but have to also consider the adversary, but also, the definition might give up a party that performs all its operations correctly, just because he is "stained" as being corrupted. Instead, here we we want to consider feasibility for the task-based natural definition of consensus discussed above, where corrupted processors are not discriminated. Clearly if the adversary can tamper or block incoming communication, it is impossible to avoid corrupted parties from outputting no output (e.g., $\perp$) or even a wrong output depending on the setup and the adversary's capabilities. For this reason, we restrict a MAd-adversary to only tamper outgoing communication. Note that this adversary might still inject messages as outgoing messages of corrupted parties.   
\end{remark}

\section{Cross-Model Reductions}

Let $M$ be any of the above models in which in the stationary case Consensus is
solvable, and in the mobile case CA is solvable.

As a simple consequence of \cite{PODC:Guerraoui00} we obtain an indulgent protocol for consensus: Take any protocol $\Pi_C$ that solve consensus in $M$ and any protocol $\Pi_{CA}$ that solves CA in the mobile analogue of $M$,  and run them alternately. A processor outputs when it commits in $\Pi_{CA}$. Nevertheless, notice that processors in our models work forever (as usually Consensus is to implement a ledger). Hence we have no notion of \emph{halt}. 

But indulgence is a motivating side show. We want to show that Consensus and CA are twins, always solvable for the same $M$ with only stationery and mobility, respectively.

One direction follows from the beginning of Section~\ref{sec:stat} below (Theorems~\ref{theorem:fitzi} and Corollary~\ref{corollary:fitzi}) and the fact that if a task is solvable in the mobile model it is also solvable in the stationary.


\cancel{\cite{}. There Garay and .... exhibit a generic template of obtaining a stationary consensus algorithm when CA is solvable. The idea is simple. Iterate the following Phase:
Start with CA. Do Rotating Coordinator among all processors interspersed with CA. The next Coordinator takes
its input from what it got in the preceding CA and send its value to all to use in a CA. Processors which committed at the previous CA keep the value they committed to as an input, the rest take their input from the Coordinator. A processor outputs, if it committed a value in the last CA in the phase.

To see why it works notice that if all started with same value, in first CA all processor will commit that value, and stop listening to the Coordinator. Otherwise, we just need to worry about agreement.
In the stationary case eventually a Coordinator whose messages are not corrupted will be encountered. If any processor has committed in the preceding CA will not listen to the Coordinator in the next round, it follows that the Coordinator value must be the value that was committed as
it was either committed or adopted by all. Thus, past the Coordinator that is not corrupted all start with the same value, ignoring the Coordinators there after.}

Thus we get a theorem:
\begin{theorem}
If CA is solvable in mobile model $M$, then Consensus is solvable in stationary $M$ .
\end{theorem}

By applying~\cite{PODC:Guerraoui00} to the above theorem, we get the following corollary.

\begin{corollary}
If CA is solvable in mobile model $M$, then Consensus is indulgent in $M$.
\end{corollary}

We would like to have the theorem: 
If Consensus is solvable in $M$,  then binary CA is solvable in mobile $M$.
We conjecture there is a way to prove this generically, but for now we show it case by case.

\section{Commit-Adopt for a Stationary Adversary}\label{sec:stat}

As a warmup we start with stationary adversasry, i.e., the adversary corrupts up to $t$ processors through the protocol and never changes his corruption. Most of the results in this section can be easily obtained by existing literature. Nonetheless, we include them here for completeness and to be able to refer to them in the following section. 

In order to establish the connection between consensus and commit-adopt we use the following simple reduction from~\cite{Fit02,TCC:ZikHauMau09} (which in terms relies on ideas from~\cite{FOCS:BerGarPer89,FitMau98}). Let \ca be a protocol for commit-adopt in the stationary setting secure against $t$ corruptions. Then the following simple  phase king protocol, which we refer to as \statcons allows us to construct consensus out of \ca. 

\begin{itemize}
    \item Let $x_i$ be the input of processor $\party_i$. Every party sets $temp_i:=x_i$
    \item For $i=1,\ldots,n$
    \begin{enumerate}
        \item The processors execute \ca on inputs  $temp_1,\ldots,temp_n$; Denote by $y_i$  the output of $\party_i$ in \ca. By definition of \ca, for each $\party_j$, for some $b_j$ we have $y_j\in\{commit(b_j),adopt(b_j)\}$  
        \item $\party_i$ sends $b_i$ to every $\party_j$ who denotes the received value as $b_{i\rightarrow j}$. 
        \item Each $\party_j$ sets
        $temp_j:=\left\{\begin{array}{l}
        b_j \text{ if } y_j=commit(b_j)\\ 
        b_{i\rightarrow j} \text{ otherwise}
        \end{array}
        \right.$
\end{enumerate}
\item Every processor outputs $temp_n$ 
\end{itemize}

\begin{theorem}[\cite{Fit02,TCC:ZikHauMau09}]\label{theorem:fitzi} If \ca solves Commit-Adopt in the stationary Byzantine MAd adversary model then \statcons solves consensus in the same model. 
\end{theorem} 

It is straight-forward to verify that the above theorem applies verbatim to send-omission corruption  and the authenticated Byzantine setting. For completeness we state this in the following corollary.

\begin{corollary}\label{corollary:fitzi}
If \ca solves Commit-Adopt in the stationary (send)-omission or Authenticated Byzantine MAd adversary model then \statcons solves consensus in the corresponding model.
\end{corollary}

Furthermore, the following theorem follows trivially from the trivial reduction of commit-adopt to consensus: 
\begin{enumerate}
\item The processors run consensus; denote by $y_i$ the output of processor $\party_i$
\item Every $\party_i$ outputs $commit(y_i)$ \end{enumerate}

\begin{theorem}\label{lemma:st_catocons}
If there exists a protocol for solving binary consensus in the stationary (send)-omission, Byzantine, or Authenticated Byzantine MAd adversary model, then there exists a protocol \ca solving binary commit-adopt in the corresponding model.
\end{theorem}

\subsection{(Send)-Omission Corruption}
A protocol for send-omission corruptions tolerating any number $t<n$ of corrupted processors follows from~\cite{TCC:ZikHauMau09}. This bound is trivially tight. We note that the analogous positive result in our model---where every party needs to output---follows directly from the corresponding bound for mobile adversary (cf. Theorem~\ref{ch2.lemma:mob_o_ca}). It following from Theorem~\ref{lemma:st_catocons} and Corollary~\ref{corollary:fitzi} that this bound is also tight for CA. 

\begin{theorem}
There exists a protocol for solving CA in the stationary (send)-omission t-MAd adversary model if and only if $t<n.$ 
\end{theorem} 


\subsection{Byzantine}

Lamport et al.~\cite{PSL80,LSP82} proved that byzantine consensus is possible if and only if $t<n/3$ of the parties are corrupted. An efficient protocol for this bound  was later given by Berman et al.~\cite{FOCS:BerGarPer89}.  We will call their protocol BGP-Consensus. Although their definition of consensus does not give any guarantees on the output of corrupted processors, one can easily use their protocol in a black-box way to add this guarantee, by adding and extra round in which every party sends his output from BGP-Consensus to everyone, and everyone outputs the value received by at least 2n/3 of the processors. Since in BGP-Consensus all uncorrupted processors output the same value $v$, all processors will receive $v$ for all of them and $v'\neq v$ from at most $t<n/3$ parties, so they will all output $v$. From the above and the equivalence of consensus and commit-adopt in the byzantine model (Theorem~\ref{lemma:st_catocons} and Corollary~\ref{corollary:fitzi}) we get the following: 

\begin{theorem}[\cite{FOCS:BerGarPer89}]\label{theorem:st_byz}
There exists a protocol for solving CA in the stationary Byzantine $t$-MAd adversary model if and only if $t<n/3.$
\end{theorem} 

\subsection{Authenticated Byzantine}
In the Authenticated Byzantine setting a lower bound of $t<n/2$ corruptions  was proved by Fitzi~\cite[Proposition~3.1]{Fit02}. 
A protocol for the authenticated setting follows from the observation that, in the stationary model, one can achieve our restrictions on the Authenticated Byzantine MAd adversary by assuming (perfectly) unforgeable signatures, and having every party digitally sign his messages using a standard existentially unforgeable signatures scheme.\footnote{Indeed, our definition of Authenticated Byzantine is equivalent to having an imaginary perfect signature scheme. Note that, in reality, such signatures do not exists, hence in an actual realization, the protocol will achieve consensus except with negligible probability. } Indeed, under this implementation, the adversary will be unable to create any message on behave of any uncorrupted processor. 

The above observation implies that the folklore reduction of consensus to broadcast for $t<n/2$ yields a consensus protocol in our model: Have every party use Dolev-Strong broadcast~\cite{DolStr83} to broadcast his input to everyone, and then take majority of the broadcasted values.   

\begin{theorem}[\cite{Fit02,DolStr83}]\label{theorem:st_auth_byz}
There exists a protocol for solving binary CA in the stationary authenticated Byzantine $t$-MAd adversary model if and only if $t<n/2.$ 
\end{theorem} 

\section{Commit-Adopt for a Mobile Adversary}

\subsection{(Send)-Omission Corruption}
We show that if in each round a MAd adversary can pick any $n-1$ processors and drop any message of this processor it wishes, and then change to another $n-1$ processors in the next round then nevertheless binary-CA is solvable.

The algorithm proceeds in two rounds: 
\begin{enumerate}
    \item Every processor sends its input to everyone
    \item Every processor $\party_i$: If there is a bit $b$ such that only $b$ was received, then send $propose$-$commit(b)$ to everyone. Else send $no$-$commit$
    \item Every processor $\party_i$: If for the bit $b$ only $propose$-$commit(b)$ was received, then output $commit(b)$; else if for a unique $b$, $commit(b)$ was received (by any party), output $adopt(b)$; otherwise output $adopt(0)$.
\end{enumerate}

In each round there is at least one processor heard by all. Thus, after one round where processors exchange inputs, a processor proposes commit $b$ in the next round only if it heard just $b$. Since all heard the same bit from some processor only a single bit can be proposed to be committed, in the next round.

In the second round, a processor that only receives proposed commit $b$, commits $b$, else if it receives proposed commit $b$, it adopts $b$.
If a processor committed then it received a proposed commit $b$ from the at least one processor all hear from. thus, if one commits all receive proposed commit, and thus will at least adopt. 

\begin{theorem}\label{ch2.lemma:mob_o_ca}
  There exists a protocol solving commit-adopt against an  mobile MAD $t$-adversary for any $t<n$. 
\end{theorem}

\subsection{Byzantine}
It follows directly from Theorem~\ref{theorem:st_byz} and the fact that a mobile adversary is at least as strong as a stationary that at most $t<n/3$ Byzantine mobile corruptions can be tolerated for commit-adopt. In the following we describe a construction that  meets this bound.  Our commit-adopt protocol follows the structure of the graded consensus construction from~\cite{Fit02,TCC:ZikHauMau09}, where we use a weak consensus primitive. 

\begin{protocol}{$\pa{\wc}(\PS,\vec{x}=(x_1,\ldots,x_n))$}
\item Each $p_i\in\PS$ sends $x_i$ to every $p_j$; $p_j$ denotes the set of
  players who sent him $0$ (resp. $1$) as $\psetp{j}{0}$ (resp.
  $\psetp{j}{1}$).
\item Each $p_j$ sets $y_j:=\left\{\begin{array}{cl} 0 &\text{ if } |\psetp{j}{0}|> 2n/3, \text{ else }\\[.5ex]
      1 &\text{ if } |\psetp{j}{1}|> 2n/3, \text{ else }\\[.5ex]
      \nv
    \end{array}\right.$
\end{protocol}

\begin{theorem}\label{ch2.lemma:perf_a.wc}
  The protocol
  \pa{\wc} satisfies the following properties against an eventually-static mobile MAD $t$-adversary with $t<n/3$: (weak consistency) There exists
  some $y\in\{0,1\}$ such that every $p_j\in \PS$ sets
  \mbox{$y_j\in\{y,\nv\}$.}  (persistency) If every $p_i\in \PS$ has the same input $x$ then they all set $y_j:=y=x$. (termination) All parties set their $y_i$ value after a single round
\end{theorem}

\begin{proof}
\medskip

\noindent{(termination) } Termination is trivial since all parties set the value of $y_i$ at the end of their single round. 
\medskip

\noindent{(weak consistency) } Assume, that a player $p_i$ sets  $y_i=0$. This means that $y_i=0: |\psetp{i}{0}|> 2n/3$. But since less than $1/3n$ of the parties in $|\psetp{i}{0}|$
  might be corrupted, this means that more than 1/3 of the parties are uncorrupted during the protocol's single round and also send $0$ to every other $p_j$. Hence $|\psetp{i}{1}|\leq 2n/3$ which means that no $p_j$ will decide on $1$. 
\medskip

\noindent{(persistency) } If all non-actively corrupted players
  have input $0$ (the case of pre-agreement on $1$ can be handled
  symmetrically) then every $p_i:\ \PS$  receives $0$ from at least all those parties, i.e., $y_i=0: |\psetp{i}{0}|> 2n/3$ and therefore outputs $0$. 
\end{proof}

\begin{fprotocol}{$\pa{\gc}(\PS,\vec{x}=(x_1,\ldots,x_n))$}
\item The players invoke $\pa{\wc}(\PS,\vec{x})$ and let $y_i$ denote the value set by $p_i$ (note that $y_i\in\{0,1,\nv\}$).
\item Each $p_i\in\PS$ sends $y_i$ to every $p_j$.  $p_j$ denotes the sets
  of players who sent him $0, 1$, and $\nv$ as $\psetp{j}{0},\psetp{j}{1}$, respectively
\item Each $p_j$ sets\\
  $b_j:=\left\{\begin{array}{ll} 0 &\text{ if } |\psetp{j}{0}|\geq|\psetp{j}{1}|
      \\
      1 &\text{ otherwise}
    \end{array}\right.$\\
 \item Each $\party_j$ output $o_j:=\left\{\begin{array}{ll} commit(b_j) &\text{ if } |\psetp{j}{z_j}|>2n/3\\
      adopt(b_j) &\text{ otherwise}
    \end{array}\right.$
\end{fprotocol}

\begin{theorem}\label{ch2.lemma:perf_a.wc}
  The protocol
  \pa{\gc} solves commit-adopt against an mobile Byzantine $t$-MAd adversary with $t<n/3$. 
\end{theorem}

\begin{proof}
We need to prove the following properties: 
  \begin{itemize}
      \item {\em (Property 1)}  If for some $z\in\{0,1\}$ some processor $\party_i\in\PS$ outputs $commit(b)$  then every processor  $\party_j\in \PS$ outputs $o_j\in\{commit(b),adopt(b)\}$. 
      \item {\em (Property 2)}  If every $p_i\in \PS$ has the same input $x$ then they output $commit(x)$.
\end{itemize}

The properties are proved in the following: 


\item{\bf Property 1.} Assume, that a player $p_i$ outputs   $o_j=commit(0)$ (the case of $o_j=commit(1)$ is handled symmetrically). This means that $\party_i$ received $0$ in the second round from more than $2n/3$ parties, i.e., $|\psetp{i}{0}|>2n/3$. Since less than $n/3$ of these parties are corrupted in Round $2$, this means that $\party_i$ have received $0$ from more than $n/3$ of the parties that were uncorrupted in Round 2, who therefore also sent $0$ to every other $p_j$. Hence in Round 2 every $p_j$ has received $0$ more than $n/3$ times, i.e., $|\psetp{i}{0}|>n/3$. Additionally, since the message sent in Round 2 is the message that is set during Round 1 (i.e., the outcome of weak consensus) the output of parties from Round 1 must have been in $\{0,\nv\}$ and therefore no party who is uncorrupted in Round 2 might sent $1$; hence, since there are at most $n/3$ corrupted parties per round, $|\psetp{j}{1}|\leq n/3<|\psetp{j}{0}|$ and every party sets $z_j=1$.

\item{\bf Property 2.} If all non-actively corrupted players
  have input $0$ (the case of pre-agreement on $1$ can be handled
  symmetrically) then by persistency of \wc everyone sets $y_i=0$ in Round 1, hence in Round 2 every uncorrupted party sends $0$ and therefore every party $\party_i$ received $0$ from at least $2/3n$ times and therefore outputs  $commit(0)$ 
\end{proof}

\subsection{Authenticated Byzantine}
Again, it follows directly from Theorem~\ref{theorem:st_auth_byz} and the fact that a mobile adversary is at least as strong as a stationary that at most $t<n/2$ Authenticated Byzantine mobile corruptions can be tolerated for commit-adopt. In the following we describe a construction that  meets this bound.  

We introduce a task that we call  $t$-MAd-Authenticated-Byzantine-SM: Every processor $p_i$ has an input value $x(i)$ from some domain $V$ and outputs a vector $\vec{y}_i=(y(1)_i,\ldots,y(n)_i)$ such that each $y(j)_i\in V\cup\{\perp\}$ and the following conditions hold: 
\begin{itemize} 
\item There exists a set of indices $Ind\subseteq[n]$, with $|Ind|\geq n-t$ such that for each $j\in Ind$, $y(j)_i=x(j)$ for all $i\in[n]$
\item if $p_i$ and $p_j$ output $y(\ell)_{i}\neq\perp$ and $y(\ell)_{j}\neq\perp$, respectively, then $y(\ell)_{i}=y(\ell)_{j}$
\end{itemize}

\begin{protocol}{$\pa{t-MAdABSM}(\PS,\vec{x}=(x(1),\ldots,x(n)))$}
\item Every $p_i$ sends $x(i)$ to every $p_j$, who denotes the received value as $v(i)_{j}$ ($v(i)_{j}$ is set to a default value, e.g., $0$, if no value was received). Let $\vec{v}_j=(v(1)_{j},\ldots,v(n)_{j})$. 
\item Every $p_i$ sends $\vec{v}_i$ to every $p_j$. $p_j$ denotes the received vector by $\vec{v}_{i\rightarrow j}=(v(1)_{{i\rightarrow j}},\ldots,v(n)_{{i\rightarrow j}})$, where $\vec{v}_{i\rightarrow j}:=\perp^n$ if  $\vec{v}_{i\rightarrow j}\not\in V^n$  was received.
\item Every $p_j$ and every $\ell\in[n]$: if for some  $b\in V$ and some set $I_i\subseteq[n]$ with $|I_i|\geq n-t$,  $v(\ell)_{{i\rightarrow j}}=b$ for all $i\in I_i$ and $v(\ell)_{{i\rightarrow j}}=\perp$ for all $i\in[n]\setminus I_i$ 
then set $y(\ell)_{j}=b$ else set $y(\ell)_{j}=\perp$. Output $\vec{y}=(y(1)_{j},\ldots,y(n)_{j})$.
\end{protocol}

\cancel{
\begin{protocol}{$\pa{\wc}(\PS,\vec{x}=(x(1),\ldots,x(n)))$}
\item[0.] Let $\rho$ denote the round in which the protocol starts.
\item Each $p_i\in\PS$ sends $(\party_i,\rho,x_i)$ to every $p_j$; $p_j$ denotes the set of
  players who sent him $(\party_i,\rho,1)$ (resp. $(\party_i,\rho,0)$ or $\perp$) as $\psetp{1,j}{1}$ (resp.
  $\psetp{1,j}{0}$).
\item Each $p_i$:  For each $p_\ell\in\psetp{i}{0}$ (resp. $p_\ell\in\psetp{i}{1}$) send 
$(\party_i,\rho+1,(\party_\ell,\rho,0))$  (resp. $(\party_i,\rho+1,(\party_\ell,\rho,1)$) to every $p_j\in\PS$; $\party_j$ denotes by $\psetp{2,i\rightarrow j}{0}$ the set of parties $\party_\ell$ such that $(\party_i,\rho+1,(\party_\ell,\rho,0))$ was received from $\party_i$ in this round ($\psetp{2,i\rightarrow j}{1}$ is defined analogously);  $\party_j$ denotes the set of parties $\party_\ell$ for which both $(\cdot,\rho+1,(\party_\ell,\rho,0))$ and $(\cdot,\rho+1,(\party_\ell,\rho,1))$ was seen (in any of the two rounds) by $\psetp{j}{01}$. For each $b\in\{0,1\}$: $\party_j$ denotes by  $\psetp{j}{b}$ the set

$$\psetp{j}{b} = \{p_k \text{ s.t. }  |(\psetp{1,j}{b}\cap\psetp{2,k\rightarrow j}{b})\setminus\psetp{j}{01}|>n/2 \}$$


\item Each $p_j$ sets $y_j:=\left\{\begin{array}{cl} 0 &\text{ if } |\psetp{j}{0}|> n/2, \text{ else }\\[.5ex]
      1 &\text{ if }   |\psetp{j}{1}|> n/2, \text{ else }\\[.5ex]
      \nv
    \end{array}\right.$
\end{protocol}
}

\begin{theorem}\label{elithm1}
  Protocol $t-MAdABSM$ solves $t$-MAd-Authenticated-Byzantine-SM in the mobile $t$-MAd  Authenticated Byzantine adversary model for $t<n/2$. 
\end{theorem}

\begin{protocol}{$\pa{\tt CA}(\PS,\vec{x}=(x(1),\ldots,x(n)))$}
\item Execute $t-MAdABSM(\PS,\vec{x})$; every processor $p_j$ denoted its output as $\vec{z}_j$. 
\item For each $p_j$: if there is a value $v\neq\perp$ such that majority of the elements in $\vec{z}_j$ equal $v$, then set $o(j):=propose.commit(v)$, else set $o(j):=propose.no.commit$. 
\item Execute $t-MAdABSM(\PS,\vec{o}=(o(1),\ldots,o(n)))$; every processor $p_j$ denotes its output as $\vec{d}_j=(d(1)_{j},\ldots,d(n)_{j})$. 
\item Every $p_j$: if for some bit $b$ the number of indices $\ell$ such that   $d(\ell)_{j}=commit(b)$ is more than $n/2$ then 
 output $commit(b)$; else if  for some bit $b'$:  $|commit(b')|>|commit(1-b')|>0$ then  $adopt(b')$; else $adopt(x(j))$.
%
%
\end{protocol}

\begin{theorem}\label{elithm2}
  Protocol $\tt CA$ solves binary CA in the mobile $t$-MAd  Authenticated Byzantine adversary model for $t<n/2$. 
\end{theorem}

\begin{proof}[Proof of Theorem~\ref{elithm1} (sketch)]
Consider processor $p_i$ which is not corrupted in line 1 (i.e. the first round).
All processor $p_j$ will have $v(i)_j= x(i)$; consequently since minority is corrupted in the second round (line 2), then majority will send $x(i)$ in the second round. On the other hand the processors corrupted in the second round cannot forge $x(i)$, hence their value for $p_i$ will be $x(i)$ or  $\bot$. 

Now let $p_i$ be a processor corrupted in the first round. To output a value different than
$\bot$ (line 3)  processor $p_j$ received same value majority $v(i)_k$ from indices $k$ whose cardinality is at least a majority. At most minority was corrupted out of this majority, thus this same value  $v(i)_k$ is the only candidate value (not $\bot$ ) to be output by any processor for the value $y(i)_j$ which is the same for all $j$.

\end{proof}

\begin{proof}[Proof of Theorem~\ref{elithm2} (sketch)]
First we show that if all start with the same value they all commit to this value. By the definition of $t$-MAd-Authenticated-Byzantine-SM, all processors will have a set $I$ with $|I|>n/2$ which return an input value which in this case they are all the same. Consequently, all processors will propose to commit this value $v$. Hence, by the property of 
$t$-MAd-Authenticated-Byzantine-SM which is used the second time, the will all output $commit(v)$ for this value. 

To complete the proof, we argue that if some processor $p_i$ outputs $commit(v)$ then everyone outputs $adopt(v)$. Notice that only a single $propose-commit(v)$ can be output in Line 2. This is because a propose commit $v$ by $p_i$ requires that in a majority of indices in its output it has value $v$. By the property of $t$-MAd-Authenticated-Byzantine-SM, no other processor output conflicts with a different value on these majority indices, then
majority for another value for any other processor is impossible.

Assume now $p_i$ committed $v$. Then it returned from line 3 of its output vector with 
majority $propose.commit(v)$. Let this majority value be $m$. Since processors have no conflict on their output entries returning from $t$-MAd-Authenticated-Byzantine-SM, it is easy to see that if another processor $p_j$ has a $\bot$
in an entry in which $p_i$ has $propose.commit(v)$, this is a result of $p_j$ having been corrupted in the first round of the $t$-MAd-Authenticated-Byzantine-SM. Thus, for each index
in which $p_i$ has $propose.commit(v)$ for $p_j$ to miss it, the adversary has to spend a corruption of the first round. Suppose $p_j$ has $q$ indices in which $p_i$ output $propose.commit(v)$, and it output $\bot$ thus it has another $t-q $ indices it can corrupt to have $propose.commit(\bar{v})$. Thus, the number of $propose.commit(\bar{v})$ is $t-q$
which the number of $propose.commit(v)$ it has is $m-q$. Since $m>t$ it will adopt $v$.

\end{proof}

\cancel{
\begin{theorem}\label{lemma:perf_ma.wc}
  The protocol
  \pa{\wc} satisfies the following properties against an eventually-static mobile MAD $t$-adversary in the authenticated setting with $t<n/2$: (weak consistency) There exists
  some $y\in\{0,1\}$ such that every $p_j\in \PS$ sets
  \mbox{$y_j\in\{y,\nv\}$.}  (persistency) If every $p_i\in \PS$ has the same input $x$ then they all set $y_j:=y=x$. (termination) All parties set their $y_i$ value after two rounds
\end{theorem}

\begin{proof}
\item (termination) Termination is trivial since all parties set the value of $y_i$ at the end of their second round. 
 \item (weak consistency) Assume, that a player $\party_i$ sets  $y_i=0$. This means that $|\psetp{i}{0}|\geq n/2$. Hence there is at least one party in $\party_{\bar{i}}\in\psetp{i}{0}$ who is honest in the second round, and therefore sent $(\party_{\bar{i}},\rho+1,(\party_\ell,\rho,0))$ for every  $\party_\ell\in\psetp{2,\bar{i}\rightarrow i}{0}$ to every party. Assume now, towards contradiction, that some $\party_j$ outputs $y_j=1$. With a similar argument as above, there exists a $\party_{\bar{j}}\in\psetp{i}{1}$ and is is honest, therefore, in the second round and therefore sent $(\party_{\bar{j}},\rho+1,(\party_\ell,\rho,1))$ for every  $\party_\ell\in\psetp{2,\bar{j}\rightarrow j}{1}$ to every party. 

Let $\psetp{1,{\bar{i}\bar{j}}}{01}$ denote the set of parties that in the first round (Round~$\rho$) sent inconsistent messages to $\party_{\bar{i}}$ and $\party_{\bar{j}}$,  i.e., $\psetp{1,{\bar{i}\bar{j}}}{01}$ is the set of parties $\party_\ell$ that sent $(\party_\ell,\rho,b)$ to $\party_{\bar{i}}$ and $(\party_\ell,\rho,1-b)$ to $\party_{\bar{j}}$ for some $b\in\{0,1\}$. This means that for some $b\in\{0,1\}$ every party $p\in\PS\setminus \psetp{1,{\bar{i}\bar{j}}}{01}$ sent $(p,\rho,b)$ to both $\party_{\bar{i}}$ and $\party_{\bar{j}}$ in the first round. Wlog assume $b=0$ (the case where $b=1$ is handled symmetrically). 
 Since both  $\party_{\bar{i}}$ and $\party_{\bar{j}}$ are honest in the second round: 
 
 \begin{equation}\label{eq:1}
 \psetp{1,{\bar{i}\bar{j}}}{01}\subseteq \psetp{i}{01}\cap\psetp{j}{01}
 \end{equation}
 
 Furthermore, since $\party_{\bar{i}}$ received $(\party,\rho,0)$ from every party in $\party\in\PS\setminus \psetp{1,{\bar{i}\bar{j}}}{01}$ we have
 
  \begin{equation}\label{eq:2}
 \PS\setminus \psetp{1,{\bar{i}\bar{j}}}{01}=  \psetp{1,{\bar{i}}}{0}\setminus \psetp{1,{\bar{i}\bar{j}}}{01}
  \end{equation}
 
  Finally, since $\party_{\bar{i}}$ is honest in the second round we have
  
   \begin{equation}\label{eq:3}
  \psetp{1,{\bar{i}}}{0} = \psetp{2,\bar{i}\rightarrow i}{0}
  \end{equation}
   
 $$|\PS\setminus \psetp{1,{\bar{i}\bar{j}}}{01}| \overset{eq.~\ref{eq:2}}{=} | \psetp{1,{\bar{i}}}{0}\setminus \psetp{1,{\bar{i}\bar{j}}}{01}| \overset{eq.~\ref{eq:3}}{=} |\psetp{2,\bar{i}\rightarrow i}{0}\setminus \psetp{1,{\bar{i}\bar{j}}}{01}|\overset{eq.~\ref{eq:1}}{>}
 |\psetp{2,\bar{i}\rightarrow i}{0}\setminus\psetp{i}{01}|>n/2$$

But since as discussed above, $\party_{\bar{j}}$ has also received $(\cdot,\rho,0)$ from all parties in $\PS\setminus \psetp{1,{\bar{i}\bar{j}}}{01}$ he must have received  $(\cdot,\rho,1)$ from less than $n/2$ parties, so he cannot be in $\psetp{j}{1}$; indeed, if $\party_{\bar{j}}\in\psetp{j}{1}$ then by definition it means that $\party_{\bar{j}}$ sent more than $n/2$ messages of the type $(\party_{\bar{j}},\rho+1,(\cdot,\rho,1))$ which contradicts the fact that $\party_{\bar{j}}$ is 
honest in the second round.



%
%
%
\item (persistency) Assume that all parties have input $0$ (the case of input $1$ is symmetric). Let $H_1$ denote the set of parties that are uncorrupted in the first round. Then in the first round, every party $\party_i$ receives $(\cdot,\rho,0)$ from more than $n/2$ parties (all the uncorrupted) which means that 
 $|H_1\cap\psetp{1,i}{0}|>n/2$. But then every party $\party'$ which is uncorrupted in the second round will send $(\party',\rho+1,(p,1,0))$ for every $p\in H_1\cap\psetp{1,i}{0}$. Furthermore, since $H_1\cap\psetp{1,i}{0}$ were honest in the first round, no party will be included in and $\psetp{j}{01}$, for any $\party_j\in\PS$, and therefore all parties that are honest in the second round will be included in every $\party_j$'s $\psetp{j}{0}$ set; since there are more than $n/2$ such parties, every $\party_j$ will output $0$. 
 \end{proof}

\begin{protocol}{$\pa{\gc}(\PS,\vec{x}=(x_1,\ldots,x_n))$}
\item[0.] Let $\rho$ denote the round in which the protocol starts.
\item The players invoke $\pa{\wc}(\PS,\vec{x})$ and let $y_i$ denote the value set by $p_i$ (note that $y_i\in\{0,1,\nv\}$).
\item (Round $\rho+2$): Each $p_i\in\PS$ sends $(\party_i,\rho+2,y_i)$ to every $p_j$; $p_j$ denotes the set of
  players who sent him $(\party_i,\rho+2,1)$ as $\psetp{1,j}{1}$, the set of
  players who sent him $(\party_i,\rho+1,\nv)$ or $\perp$ as $\psetp{1,j}{\nv}$, and the set of parties that sent $(\party_i,1,0)$ or as $\psetp{1,j}{0}$.

  \item Each $\party_i$:  For each $p_\ell\in\psetp{i}{0}$ (resp. $p_\ell\in\psetp{i}{1}$) send 
$(\party_i,\rho+3,(\party_\ell,\rho+2,0))$  (resp. $(\party_i,\rho+3,(\party_\ell,\rho+2,1)$) to every $p_j\in\PS$; $\party_j$ denotes by $\psetp{2,i\rightarrow j}{0}$ the set of parties $\party_\ell$ such that $(\party_i,\rho+3,(\party_\ell,\rho+2,0))$ was received from $\party_i$ in this round ($\psetp{2,i\rightarrow j}{1}$ and $\psetp{2,i\rightarrow j}{\nv}$ are defined analogously);  $\party_j$ denotes the set of parties $\party_\ell$ for which both $(\cdot,\rho+3,(\party_\ell,\rho+2,0))$ and $(\cdot,\rho+3,(\party_\ell,\rho+2,1))$ was seen (in any of the two rounds) by $\psetp{j}{01}$. For each $b\in\{0,1\}$: $\party_j$ denotes by  $\psetp{j}{b}$ the set

$$\psetp{j}{b} = \{p_k \text{ s.t. }  |(\psetp{1,j}{b}\cap\psetp{2,k\rightarrow j}{b})\setminus\psetp{j}{01}|>n/2 \}$$
  
 \item Each $\party_j$ sets 
 
$b_j:=\left\{\begin{array}{cl} 1 &\text{ if } |\psetp{1,j}{1}\setminus\psetp{j}{01} |> |\psetp{1,j}{0}\setminus\psetp{j}{01} |, \text{ else }\\[.5ex]
0
    \end{array}\right.$


\item Each $p_j$ outputs $o_j:=\left\{\begin{array}{cl} commit(b_j) &\text{ if } |\psetp{j}{b}|> n/2, \text{ else }\\[.5ex]
adopt(b_j)
    \end{array}\right.$
\end{protocol}

\begin{theorem}\label{lemma:perf_ma.wc}
  The protocol
  \pa{\gc} solves commit-adopt against an eventually-static mobile MAD $t$-adversary with $t<n/3$. 
\end{theorem}

\begin{proof}
We need to prove the following properties: 
  \begin{itemize}
      \item {\em (Property 1)}  If for some $z\in\{0,1\}$ some processor $\party_i\in\PS$ outputs $commit(b)$  then every processor  $\party_j\in \PS$ outputs $o_j\in\{commit(b),adopt(b)\}$. 
      \item {\em (Property 2)}  If every $p_i\in \PS$ has the same input $x$ then they output $commit(x)$.
\end{itemize}

\item{\bf Property 1.}  

Assume wlog that some party $\party_i$ outputs $o_i=commit(0)$ (the case $o_i=commit(1)$ is symmetric).  This means that there are more than $n/2$ parties $\party_k$ for which $|\psetp{2,k\rightarrow j}{0}|>n/2$. But then one of these parties, say $\party_{k*}$ must be honest in the last round (i.e., Round $\rho_1$) and therefore $\psetp{2,k\rightarrow j}{0}$ includes the actual bit $0$ that he received in the previous round, i.e., Round $\rho+2$. Hence,  $\party_{k*}$ must have received $(\party,\rho_2,0)$ from at least one honest party $\party_j$, which means that the value $y_i$ that $\party_j$ set in the end of $\wc$ must be $y_j=0$. But in that case, Weak Consistency (Lemma~\ref{lemma:perf_ma.wc}) ensures that all parties $\party_\ell$ set $y_j\in\{0,\nv\}$. We next note that since $\party_{k*}$ is uncorrupted in Round $\rho+3$ : 

$$
|\psetp{2,k*\rightarrow \ell}{0}|>n/2
$$

for every $\party_{\ell}\in\PS$. Assume towards contradiction that some $\party_\ell$ outputs $o_\ell\in\{commit(1),adopt(1)\}$. Clearly, all parties in $\psetp{1}{1,\ell}$ must be corrupted in Round $\rho+2$. 

 Next we note that since $\party_{k*}$ is uncorrepted in Round $\rho+3$, $\psetp{2,k*\rightarrow \ell}{0}$ includes all messages received in round $\rho+2$. Furthermore, again since $\party_{k*}$ is uncorrupted in Round $\rho+3$ and all uncorrupted parties enter with input  $0$ or $\nv$ all parties in $(\psetp{1,\ell}{1}\cup\psetp{1,\ell}{\nv}\cup_{q=1}^{n}(\psetp{2,q\rightarrow \ell}{1}\cup\psetp{2,q\rightarrow \ell}{\nv})$ mush have be corrupted in Round $\rho+2$. 

%

 Now we observe that: 
  
 
$$\psetp{2,k^*\rightarrow\ell}{0}= \underbrace{\left(\psetp{2,k^*\rightarrow\ell}{0}\cap\left(\psetp{1,\ell}{0}\cap_{q=1}^{n}\psetp{2,q\rightarrow \ell}{0}\right)\right)}_{\psetp{1,\ell}{0}\setminus\psetp{01}{\ell}}\bigcup\underbrace{\left(
\psetp{2,k^*\rightarrow\ell}{0}\bigcap\left(\psetp{1,\ell}{1}\cup\psetp{1,\ell}{\nv}\cup_{q=1}^{n}(\psetp{2,q\rightarrow \ell}{1}\cup\psetp{2,q\rightarrow \ell}{\nv}\right)\right)}_{T_2}$$

Note that $T_2\subseteq\psetp{\ell}{01}$. Let $|T_2|=t_2$. By inspection of the protocol the above have:

 $$|\psetp{1,\ell}{0}\setminus\psetp{\ell}{01}|=|\psetp{2,k^*\rightarrow\ell}{0}| - |T_2|>n/2-t_2$$
 
 Since all the parties in $\psetp{1}{\ell}$ must be corrupted in round $\rho+2$ (from the weak consistency of \wc): 
 
  $$|\psetp{1,\ell}{1}\setminus\psetp{\ell}{01}|\leq t-t_2$$, where $t$ denotes  the number of corrupted parties in round $\rho+2$. From the above we have:  
  
  $$|\psetp{1,\ell}{0}\setminus\psetp{\ell}{01}|>|\psetp{1,\ell}{1}\setminus\psetp{\ell}{01}|
  $$
  
  so $\party_\ell$ cannot have output  $o_j\in\{commi(1),adopt(1)\}$. 

\item{\bf Property 2.} If every party has input $b$, then by Lemma~\ref{lemma:perf_ma.wc}, all parties output $b$ at the end of the second round. The fact that every party $\party_j$ will then output $commit(b)$ follows directly from the analysis of the persistency from~\ref{lemma:perf_ma.wc} since the exact same condition is used. 
\end{proof}
}

\section{Conclusions}

We have examined the Stationary/Mobile replacing Synchronous/Asynchronous and shown that the former when considered in the MAd adversary model is indulgent for common models unlike
its counterpart.

We contend that there is a much richer clean distributed Theory when we consider MAd adversary in the context of Stationary/Mobile as conjectured in the paper. Showing that our case by case analysis was superfluous is a beautiful challenge. Same goes for changing binary CA to multi-valued CA in the mobile benign omission case.

Is beautiful theory with no current application worth developing? It is called Basic Research (BR) and we still root for BR.

In~\cite{DBLP:journals/corr/DolevG16}, Dolev and Gafni analyse mixtures of Stationary and mobiles faults. It is interesting the re-examine \cite{} in light of this paper.

Finally, it will be nice to identify other natural ``pairs'' for which indulgence can be defined.

\bibliographystyle{plain}
\bibliography{biblio}

\end{document}